\documentclass[review]{elsarticle}
\usepackage[letterpaper,top=2cm,bottom=2cm,left=3cm,right=3cm,marginparwidth=1.75cm]{geometry}
\usepackage{lineno,hyperref,mathtools}

\journal{Journal of \LaTeX\ Templates}

\usepackage{amssymb}
\usepackage{bm}
\usepackage{amssymb}
\usepackage{amsthm}
\usepackage{multirow,booktabs}
\usepackage{cases}
\usepackage{threeparttable}

\newtheorem{theorem}{Theorem}

\newtheorem{lemma}{Lemma}
\newtheorem{corollary}{Corollary}
\newtheorem{example}{Example}

\newcommand{\C}{{\mathcal{C}}}
\newcommand{\F}{{\mathbb{F}}}

\begin{document}

\begin{frontmatter}

\title{Quantum MDS Codes with length $n\equiv 0,1($mod$\,\frac{q\pm1}{2})$}
\tnotetext[mytitlenote]{This research work is supported by the National Natural Science Foundation of China under Grant Nos. U21A20428 and 12171134.}

\author[mymainaddress]{Ruhao Wan}
\ead{wanruhao98@163.com}

\address[mymainaddress]{School of Mathematics, HeFei University of Technology, Hefei 230601, China}

\begin{abstract}
An important family of quantum codes is the quantum maximum-distance-separable (MDS) codes.
In this paper,
we construct some new classes of quantum MDS codes by generalized Reed-Solomon (GRS) codes and Hermitian construction.
In addition, the length $n$ of most of the quantum MDS codes we constructed satisfies $n\equiv 0,1($mod$\,\frac{q\pm1}{2})$, which is different from previously known code lengths.
At the same time, the quantum MDS codes we construct have large minimum distances that are greater than $q/2+1$.
\end{abstract}


\end{frontmatter}

\section{Introduction}\label{sec1}

Quantum codes play an important role in quantum communication and quantum computing.
It is well known that well-constructed quantum codes are of great theoretical significance.
In \cite{RefJ (1998) introduct 1}, Calderbank et al. showed for the first time that
the construction of a quantum code can be reduced to the construction of a classical linear code with certain self-orthogonality properties over $\F_2$ or $\F_4$.
Then Rains \cite{RefJ (1999) introduct 2} and Ashikhmin et al. \cite{RefJ (2001) introduct 3} generalized to the non-binary case.
Afterward, the researchers gave many methods to obtain quantum codes with good parameters (see \cite{RefJ (2006) introduct 6}-\cite{RefJ (2015) n=q^2+1(2)}).
Later, related research gradually became a hot issue.

Let $q$ be a prime power.
We always use $[[n,k,d]]_q$ to denote a quantum code of length $n$, size $q^k$
and minimum distance $d$.
A quantum code with minimum distance $d$ can correct any $\lfloor \frac{d-1}{2}\rfloor$ quantum errors and detect any $d-1$ quantum errors.
Like classical codes, quantum codes with parameters $[[n,k,d]]_q$ need to satisfy the quantum Singleton bound:  $2d\leq n-k+2$ (see \cite{RefJ (2001) introduct 3}).
A quantum code achieving the above bound is called a quantum maximum-distance-separable (MDS) code.
As mentioned in the literature \cite{RefJ (2010) L.jin}, the minimum distance for almost all known $q$-ary quantum MDS codes is $d\leq q/2+1$, except for a few lengths.
Therefore, constructing quantum MDS codes with minimum distance $d>q/2+1$ is a valuable and interesting task.

\subsection{Related works}

In recent years,
the construction of quantum MDS codes has been extensively studied.
In \cite{RefJ (2004) n<q+1(1),RefJ (2004) n<q+1(2)}, the problem of constructing $q$-ary quantum MDS codes with $n\leq q+1$ has been solved.
There are two main ways to construct quantum MDS codes with good parameters, namely using constacyclic codes and GRS codes.
On one hand, in \cite{RefJ (2013) n=q^2+1(3),RefJ (2014) kai}, Kai et al. constructed several classes of quantum MDS codes by constacyclic codes and negacyclic codes.
Since then, researchers have constructed many quantum MDS codes  via
constacyclic codes (see \cite{RefJ (2014) Zhang}-\cite{RefJ (2015) B.chen}) and pseudo-cyclic codes (see \cite{RefJ (2016) S.Li}).
On other hand,
GRS codes are also widely used to construct quantum MDS codes due to their good algebraic properties.
In \cite{RefJ (2008) n=q^2+1(4)},
A unified framework for constructing quantum MDS codes via GRS codes was proposed by Li et al.
Then in \cite{RefJ (2012) L.jin,RefJ (2014) L.jin}, Jin et al. generalized the conclusion in \cite{RefJ (2008) n=q^2+1(4)} and constructed several classes of quantum MDS codes.
Since then, researchers have constructed many quantum MDS codes with distance greater than $q/2+1$ using GRS codes
(see \cite{RefJ (2016) X.He}-\cite{RefJ (2024) Fang}).
In particular, recently in \cite{RefJ (2021) n=q^2+1(5)}, Ball proved that the minimum distance of a quantum MDS code constructed by GRS is at most $q+1$.
We summarize some known results in Table \ref{tab:1} of Section \ref{sec5}.

\subsection{Our results}

In this paper, we construct some new classes of quantum MDS codes by Hermitian self-orthogonal GRS codes.
By selecting the suitable code locators and column multipliers, we construct several new classes of Hermitian self-orthogonal GRS codes, and thus we obtain several new classes of quantum MDS codes (see Theorems \ref{th 11}, \ref{th 22}, \ref{th 33}, \ref{th 44}).
We can find that the length of these quantum MDS codes satisfies $n\equiv 0,1($mod$\,\frac{q\pm1}{2})$, which is a length not considered in the literature.
The quantum MDS codes we construct are new and have large minimum distances.
We list our new constructions in Table \ref{tab:3} of Section \ref{sec5}.

\subsection{Organization of this paper}

The paper is organized as follows.
In Section \ref{sec2}, we briefly review the definitions of Hermitian self-orthogonality and GRS codes, and introduce some main lemmas.
In Section \ref{sec4}, we construct some new classes of quantum MDS codes.
In Section \ref{sec5}, we illustrate by comparison that our quantum MDS code has new parameters.
Section \ref{sec6} concludes the paper.

\section{Preliminaries}\label{sec2}

In this section, we introduce some basic definitions and useful results on Hermitian self-orthogonality and GRS codes, and introduce some main lemmas.
Let $\F_{q^2}$ be the finite field with $q^2$ elements, where $q$ is an odd prime power.
A linear code $\C$ of length $n$ over $\F_{q^2}$ is said to be an $[n,k,d]_{q^2}$ code if its
dimension is $k$ and its minimum distance is $d$.

\subsection{Hermitian self-orthogonality}

For any two vectors $\bm{x}=(x_1,x_2,\dots,x_n),\ \bm{y}=(y_1,y_2,\dots,y_n)\in \F_{q^2}^n$,
the Euclidean and Hermitian inner product of vectors $\bm{x}$, $\bm{y}$ are defined as
$$\langle \bm{x},\bm{y} \rangle=\sum_{i=1}^{n}x_iy_i\quad and\quad \langle \bm{x},\bm{y} \rangle_{H}=\sum_{i=1}^{n}x_iy_i^q,$$
respectively.
Let $\C$ be a linear code of length $n$ over $\F_{q^2}$. Then
$$\C^{\perp}=\{\bm{x}\in \F_{q^2}^n:\langle \bm{x}, \bm{y}\rangle=0, \ for \ all\ \bm{y}\in \C \}\quad and \quad
\C^{\perp_{H}}=\{\bm{x}\in \F_{q^2}^n:\langle \bm{x}, \bm{y}\rangle_{H}=0, \ for \ all\ \bm{y}\in \C \}$$
are called the the Euclidean and Hermitian dual code of $\C$, respectively.
If $\C\subseteq \C^{\perp_H}$, the code $\C$ is called Hermitian self-orthogonal.
For a vector $\bm{x}=(x_1,x_2,\dots,x_n)\in \F_{q^2}^n$,
let $\bm{x}^i=(x_1^i,x_2^i,\dots,x_n^i)$.
Specially, $0^0=1$.
For a linear code $\C$ over $\F_{q^2}$ with a generator matrix $G$,
we have $\C^{\bot_H}=(\C^{(q)})^{\bot}$, where $\C^{(q)}=\{\bm{c}^q:\bm{c}\in \C\}$.
For a matrix $A=(a_{ij})$, let $A^{(q)}$ be the matrix $(a_{ij}^q)$,
then we can obtain that $G^{(q)}$ is a generator matrix of $\C^{(q)}$.

\subsection{Generalized Reed-Solomon codes}

Choose two vectors $\bm{a}=(a_{1},a_{2},\dots,a_{n})$ and $\bm{v}=(v_1,v_2,\dots,v_n)$,
where $a_i$ are distinct elements in $\F_{q^2}$ and $v_i\in \F_{q^2}^*=\F_{q^2}\setminus \{0\}$ ($v_i$ may not be distinct).
For an integer $1\leq k\leq n$,
the GRS code associated with $\bm{a}$ and $\bm{v}$ is defined as
$$GRS_{k}(\bm{a},\bm{v})=\{(v_{1}f(a_{1}),v_{2}f(a_{2}),\dots,v_{n}f(a_{n})):f(x)\in \F_{q^2}[x],\ \deg(f(x))\leq k-1\}.$$
We can know that the code $GRS_{k}(\bm{a},\bm{v})$ is an MDS code with parameters $[n,k,n-k+1]_{q^2}$.
And its dual code is also an MDS code.
The elements $a_1, a_2,\dots,a_n$ are called the code locators of
$GRS_k(\bm{a},\bm{v})$,
and the elements $v_1,v_2,\dots,v_n$ are called the column multipliers of $GRS_k(\bm{a},\bm{v})$.

\subsection{Main Lemmas}

In \cite{RefJ (2014) L.jin,RefJ (2017)Lem GRS}, a necessary and sufficient condition for a GRS code to be Hermitian self-orthogonal is given.

\begin{lemma}\label{lem GRS panding}(\cite{RefJ (2014) L.jin,RefJ (2017)Lem GRS})
The two vectors $\bm{a}$ and $\bm{v}$
are defined as above.
Then $GRS_k(\bm{a},\bm{v})\subseteq GRS_k(\bm{a},\bm{v})^{\bot_H}$ if and only if $\langle \bm{a}^{qi+j}, \bm{v}^{q+1} \rangle=0$, for all $0\leq i,j \leq k-1$.
\end{lemma}

In \cite{RefJ (2001) introduct 3}, Ashikhmin et al. gave a method to construct quantum codes by classical codes.

\begin{lemma}\label{Lem zhuyao}
(Hermitian construction \cite{RefJ (2001) introduct 3})
If there exists an $[n,k,n-k+1]_{q^2}$-Hermitian self-orthogonal code,
then there exists an $[[n,n-2k,k+1]]_{q}$-quantum code.
\end{lemma}

Given a quantum MDS code, we can obtain a new quantum code by the following lemma.

\begin{lemma}\label{lem propagation rule}
(Propagation Rule \cite{RefJ (2015) n=q^2+1(2)})
If there exists an $[[n,n-2k,k+1]]_q$-quantum MDS code,
then there exists an $[[n-1,n-2k+1,k]]_{q}$-quantum MDS code.
\end{lemma}

The next two lemmas give sufficient conditions for a system of linear equations to have a solution.

\begin{lemma}\label{lem two youjie 1}(\cite{RefJ (2023) R. Wan})
Suppose that $A$ is a $t\times n$ matrix over $\F_{q}$, where $1\leq t< n< q+1$.
Let $A_i$ be the $t\times (n-1)$ matrix obtained by deleting the $i$-th column from $A$.
If $rank(A)=rank(A_1)=\dots=rank(A_n)$,
then the following equation
$$A\bm{u}^T=\bm{0}^T$$
has a solution $\bm{u}=(u_1,u_2,\dots,u_{n})\in (\F_q^*)^{n}$.
\end{lemma}

\begin{lemma}\label{lem two youjie 2}(\cite{RefJ (2023) R. Wan})
Suppose that $A$ is a $t\times n$ matrix over $\F_{q^2}$, where $1\leq t\leq n< q+1$.
If the following conditions are met:
\begin{itemize}
\item[(1)] the equation $A\bm{u}^T=\bm{0}^T$ has a solution $\bm{u}=(u_1,u_2,\dots,u_{n})\in (\F_{q^2}^*)^{n}$;
\item[(2)]$A^{(q)}$ is row equivalent to $A$,
\end{itemize}
then the following equation
$$A\bm{u}^T=\bm{0}^T$$
has a solution $\bm{u}=(u_1,u_2,\dots,u_{n})\in (\F_q^*)^{n}$.
\end{lemma}

\section{The construction of quantum MDS codes}\label{sec4}

In this section,
by  Hermitian self-orthogonal GRS codes we construct quantum MDS codes of length satisfying $n\equiv 0,1($mod$\,\frac{q\pm1}{2})$.

We assume that $q$ is an odd prime power and $q^2-1=2hm$.
For brevity,
let $\theta$ be a primitive element of $\F_{q^2}$
and $\gamma=\theta^{2h}$, then $\gamma$ is a primitive $m$-th root of unity.
Let $\alpha=\theta^{m}$ and $\beta=\theta^{2m}$, then we have $\alpha^{v}\neq \alpha^{v'}$ for $1\leq v\neq v'\leq 2h$, and $\beta^{u}\neq \beta^{u'}$ for $1\leq u\neq u'\leq h$.
Denote
$(\langle \gamma\rangle)=(1,\gamma,\dots,\gamma^{m-1})$
and $\bm{1}_n$ be the all-one row vector with length $n$.
Assume $\xi=\theta^{-\frac{q+1}{2}}$, then we have $\xi^q=-\xi$.
For ease of representation, we define $[\ ]$ as follows,
\[[a,b]=\begin{cases}
\emptyset, & if\ a>b; \\
\{a,b\}, & if\  a\leq b.
\end{cases}\]

\subsection{quantum MDS codes of length $n=r\frac{q^2-1}{2h}+1$, where $\frac{q-1}{h}$ is odd}

Throughout this section,
let $h$ be an even integer such that $\frac{q-1}{h}=2\tau+1$ for some $\tau \geq 1$.

\begin{lemma}\label{lem sec 4 1}
Let $q$, $h$, $\tau$ and $m$ be defined as above.
Suppose $1\leq k\leq (h+2t+1)\frac{q-1}{2h}+\frac{1}{2}$, where $0\leq t\leq \frac{h}{2}-1$.
Then for $0\leq i,j\leq k-1$, $qi+j=sm$ if and only if
$s\in \{0\}\cup[1,3,\dots,2t-1]\cup[h+1,h+3,\dots,h+2t-1]\cup \{2, 4,\dots, h+2t\}$.
\end{lemma}

\begin{proof}
Note that $0\leq i,j< q-1$,
we have $0\leq qi+j< q^2-1$.
If $qi+j=sm=s\frac{q^2-1}{2h}$,
then $0\leq s< 2h$.
Now, we divide our proof into the following two parts according to whether $s$ is even or odd.

$\bullet$ $\textbf{Case 1:}$ $s$ is odd.

For $1\leq s\leq h-1$,
we have
\[\begin{split} qi+j&
=(\frac{s-1}{2}\frac{q-1}{h})q+\frac{q+1}{2}\frac{q-1}{h}+\frac{s-1}{2}\frac{q-1}{h}\\
                           &=(\frac{s-1}{2}\frac{q-1}{h}+\tau)q+\frac{q+1}{2}+\frac{s-1}{2}\frac{q-1}{h}+\tau\\
                          &=(\frac{s}{2}\frac{q-1}{h}-\frac{1}{2})q+\frac{h+s}{2}\frac{q-1}{h}+\frac{1}{2}.\\
	\end{split}\]
It follows that
$$i=\frac{s}{2}\frac{q-1}{h}-\frac{1}{2},\quad j=\frac{h+s}{2}\frac{q-1}{h}+\frac{1}{2}.$$

If $2t+1\leq s\leq h-1$, then
$$j=\frac{h+s}{2}\frac{q-1}{h}+\frac{1}{2}\geq (h+2t+1)\frac{q-1}{2h}+\frac{1}{2}\geq k,$$
which contradicts to the fact $j\leq k-1$;

For $h+1\leq s\leq 2h-1$,
we have
\[\begin{split} qi+j&=(\frac{s-1}{2}\frac{q-1}{h}+\tau+1)q+\frac{q+1}{2}+\frac{s-1}{2}\frac{q-1}{h}+\tau-q\\
                           &=(\frac{s}{2}\frac{q-1}{h}+\frac{1}{2})q+\frac{s-h}{2}\frac{q-1}{h}-\frac{1}{2}.\\
	\end{split}\]
It follows that
$$i=\frac{s}{2}\frac{q-1}{h}+\frac{1}{2},\quad j=\frac{s-h}{2}\frac{q-1}{h}-\frac{1}{2}.$$

If $h+2t+1\leq s\leq 2h-1$, then
$$i=\frac{s}{2}\frac{q-1}{h}+\frac{1}{2}\geq (h+2t+1)\frac{q-1}{2h}+\frac{1}{2}\geq k,$$
which contradicts to the fact $i\leq k-1$.
Hence, $s\in [1,3,\dots,2t-1]\cup[h+1,h+3,\dots,h+2t-1]$.

$\bullet$ $\textbf{Case 2:}$ $s$ is even.

For $0\leq s\leq 2h-2$, we have
$$qi+j=\frac{s}{2}\frac{q-1}{h}q+\frac{s}{2}\frac{q-1}{h}.$$
It follows that $i=j=\frac{s(q-1)}{2h}$.
If $s\geq h+2t+2$, then
$$i=j=\frac{s}{2}\frac{q-1}{h}\geq (h+2t+2)\frac{q-1}{2h}\geq (h+2t+1)\frac{q-1}{2h}+\frac{1}{2}\geq k,$$
which contradicts to the fact $i\leq k-1$.
Hence, $s\in \{0,2, 4,\dots, h+2t\}$.
This completes the proof.
\end{proof}

\begin{theorem}\label{th 11}
Let $q$ be an odd integer and $n=\frac{r(q^2-1)}{2h}+1$, where $\frac{q-1}{h}=2\tau+1$ for some $\tau\geq 1$.
The following statements hold.
\begin{itemize}
\item[(1)] For $\frac{h}{2}+1\leq r\leq h$ and $1\leq k\leq (h+1)\frac{q-1}{2h}+\frac{1}{2}$,
    there exists a quantum MDS code with parameters $[[n,n-2k,k+1]]_q$.
\item[(2)] For odd $h< r< 2h$ and $1\leq k\leq r\frac{q-1}{2h}+\frac{1}{2}$,
  there exists a quantum MDS code with parameters $[[n,n-2k,k+1]]_q$.
\end{itemize}
\end{theorem}

\begin{proof}
Let
$$\bm{a}=(0,\theta^{i_1}(\langle \gamma\rangle),\theta^{i_2}(\langle \gamma\rangle),\dots,\theta^{i_r}(\langle \gamma\rangle))\in \F_{q^2}^{r\frac{q^2-1}{2h}+1},$$
where $i_1,i_2,\dots,i_r$ are distinct modulo $2h$ and
$$\bm{v}=(v_0,v_1\bm{1}_{\frac{q^2-1}{2h}},v_2\bm{1}_{\frac{q^2-1}{2h}},\dots,v_r\bm{1}_{\frac{q^2-1}{2h}})\in (\F_{q^2}^*)^{r\frac{q^2-1}{2h}+1},$$
where $v_0,v_1,\dots,v_r\in \F_{q^2}^*$.
Then when $(i,j)=(0,0)$,
we have
$$\langle \bm{a}^{0}, \bm{v}^{q+1}\rangle=v_0^{q+1}+m\sum_{l=1}^{r}v_l^{q+1}.$$
And when $(i,j)\neq (0,0)$,
we have
$$\langle \bm{a}^{qi+j},\bm{v}^{q+1}\rangle=\sum_{l=1}^r\theta^{i_l(qi+j)}v_l^{q+1}\sum_{\nu=0}^{m-1}\gamma^{\nu(qi+j)},$$
thus
\[\langle \bm{a}^{qi+j},\bm{v}^{q+1}\rangle=\begin{cases}
0, & if\ m\nmid (qi+j); \\
m\sum_{l=1}^r\theta^{i_l(qi+j)}v_l^{q+1}, & if\  m\mid (qi+j).
\end{cases}\]

(1) Since  $1\leq k\leq (h+1)\frac{q-1}{2h}+\frac{1}{2}$, where $t=0$,
by Lemma \ref{lem sec 4 1}, when $m\mid (qi+j),\ (i,j)\neq (0,0)$,
we have
$\langle \bm{a}^{qi+j},\bm{v}^{q+1}\rangle =m\sum_{l=1}^r\alpha^{i_ls}v_l^{q+1}$, where $s\in \{2,4,\dots, h\}$.
Since $r\leq h$, we can assume that $i_1,i_2,\dots,i_r$ are distinct modulo $h$.
Note that $\beta^q=\beta$, then $\beta\in \F_{q}$.
Let
$$A=\begin{pmatrix}
 1&  1    &  1  &  \dots  & 1  \\
 0& \beta^{i_1} &  \beta^{i_2} & \dots & \beta^{i_r}\\
 \vdots & \vdots &  \vdots & \ddots & \vdots\\
 0& \beta^{i_1\frac{h}{2}} &  \beta^{i_2\frac{h}{2}} & \dots & \beta^{i_r\frac{h}{2}}\\
\end{pmatrix}$$
be a $(1+\frac{h}{2})\times (r+1)$ matrix over $\F_{q}$.
Let $A_i\ (1\leq i\leq r+1)$ be the $(1+\frac{h}{2})\times r$ matrix obtained by deleting the $i$-th column from $A$.
We can get $rank(A)=rank(A_1)=\dots=rank(A_{r+1})=1+\frac{h}{2}$.
By Lemma \ref{lem two youjie 1}, the equation $A\bm{u}^T=\bm{0}^{T}$ has a solution $\bm{u}=(u_0,u_1,\dots,u_r)\in (\F_{q}^*)^{r+1}$.

(2)
Since $1\leq k\leq (h+2t+1)\frac{q-1}{2h}+\frac{1}{2}$, where $t=\frac{r-h-1}{2}$,
by Lemma \ref{lem sec 4 1},
when $m\mid (qi+j),\ (i,j)\neq (0,0)$, we have
$\langle \bm{a}^{qi+j},\bm{v}^{q+1}\rangle =m\sum_{l=1}^r\alpha^{i_ls}v_l^{q+1},$
where $s\in [1,3,\dots,2t-1]\cup [h+1,h+3,\dots,h+2t-1]\cup \{2, 4,\dots, h+2t\}$.
Let
$$B=\begin{pmatrix}
 \alpha^{i_1}  &  \alpha^{i_2}  & \dots & \alpha^{i_r}\\
 \alpha^{3i_1}  &  \alpha^{3i_2}  & \dots & \alpha^{3i_r}\\
  \vdots & \vdots &  \ddots & \vdots\\
 \alpha^{(2t-1)i_1}  &  \alpha^{(2t-1)i_2}  & \dots & \alpha^{(2t-1)i_r}\\
 \alpha^{(h+1)i_1}  &  \alpha^{(h+1)i_2}  & \dots & \alpha^{(h+1)i_r}\\
 \alpha^{(h+3)i_1}  &  \alpha^{(h+3)i_2}  & \dots & \alpha^{(h+3)i_r}\\
  \vdots & \vdots &  \ddots & \vdots\\
 \alpha^{(h+2t-1)i_1}  &  \alpha^{(h+2t-1)i_2}  & \dots & \alpha^{(h+2t-1)i_r}\\
\end{pmatrix},$$
$$
C=\begin{pmatrix}
   \alpha^{2i_1}  &  \alpha^{2i_2}  & \dots & \alpha^{2i_r}\\
    \alpha^{4i_1}  &  \alpha^{4i_2}  & \dots & \alpha^{4i_r}\\
  \vdots & \vdots &  \ddots & \vdots\\
  \alpha^{(h+2t)i_1}  &  \alpha^{(h+2t)i_2}  & \dots & \alpha^{(h+2t)i_r}\\
\end{pmatrix}
\quad and\quad A=\begin{pmatrix}
1 & \bm{1}_r\\
\bm{0}^T & B\\
\bm{0}^T & C\\
\end{pmatrix}
$$
be the matrices over $\F_{q^2}$, respectively.
Let
$$\tilde{A}=\begin{pmatrix}
1& 1& 1& \dots &1\\
0& \alpha^{i_1}  &  \alpha^{i_2}  & \dots & \alpha^{i_r}\\
0&  \alpha^{2i_1}  &  \alpha^{2i_2} \ & \dots & \alpha^{2i_r} \\
\vdots &    \vdots & \vdots &  \ddots & \vdots\\
  0&   \alpha^{(h+2t)i_1} &  \alpha^{(h+2t)i_2} & \dots & \alpha^{(h+2t)i_r}\\
\end{pmatrix}$$
be an $r\times (r+1)$ matrix over $\F_{q^2}$.
Let $\tilde{A}_i\ (1\leq i\leq r+1)$ be the $r\times r$ matrix obtained by deleting the $i$-th column from $\tilde{A}$.
We can get $rank(\tilde{A})=rank(\tilde{A}_1)=\dots=rank(\tilde{A}_{r+1})=r$.
By Lemma \ref{lem two youjie 1}, the equation $\tilde{A}\bm{u}^T=\bm{0}^{T}$ has a solution $\bm{u}=(u_0,u_1,\dots,u_r)\in (\F_{q^2}^*)^{r+1}$.
Then the equation $A\bm{u}^T=\bm{0}^{T}$ has a solution $\bm{u}=(u_0,u_1,\dots,u_r)\in (\F_{q^2}^*)^{r+1}$.
Note that $\alpha^h=-1$, $\alpha^q=-\alpha$ and $\alpha^{2h}=1$, then for any odd number $1\leq s\leq 2t-1$, and for any $1\leq j\leq r$, we have
\[\begin{split} \alpha^{si_jq}&=(-\alpha^s)^{i_j}\\
                           &=\alpha^{(h+s)i_j}.\\
	\end{split}\]
And for any even number $2\leq s\leq h+2t$, and for any $1\leq j\leq r$, we have
$\alpha^{si_jq}=\alpha^{si_j}$.
It follows that $A$ is row equivalent to $A^{(q)}$.
By Lemma \ref{lem two youjie 2}, the equation $A\bm{u}^T=\bm{0}^{T}$ has a solution $\bm{u}=(u_0,u_1,\dots,u_r)\in (\F_{q}^*)^{r+1}$.

Combine (1) and (2). Let $v_l^{q+1}=u_l$ for $1\leq l\leq r$ and let $v_0\in \F_{q^2}^*$ such that $v_0^{q+1}=u_0m$.
We have
$$\langle \bm{a}^{0},\bm{v}^{q+1}\rangle=v_0^{q+1}+m\sum_{l=1}^rv_l^{q+1}=m\sum_{l=0}^{r}u_l=0.$$
And when $m\mid qi+j$, we have
$$\langle \bm{a}^{qi+j},\bm{v}^{q+1}\rangle=m\sum_{l=1}^{r}\alpha^{i_ls}v_l^{q+1}=m\sum_{l=1}^{r}\alpha^{i_ls}u_l=0.$$
Hence $\langle \bm{a}^{qi+j},\bm{v}^{q+1}\rangle=0$, for all $0\leq i,j\leq k-1$.
By Lemmas \ref{lem GRS panding} and \ref{Lem zhuyao}, we have the conclusion.
This completes the proof.
\end{proof}

\begin{example}
By taking $h=8,12$ in Theorem \ref{th 11}, we get the following quantum MDS codes.
\begin{itemize}
\item[(1)]
When $\frac{q-1}{8}=2\tau+1$ in Theorem \ref{th 11} (1),
then there exists a quantum MDS code with parameters
$[[\frac{r}{16}(q^2-1)+1,\frac{r}{16}(q^2-1)-2k+1,k+1]]_q$
for any $1\leq k\leq \frac{9(q-1)}{16}+\frac{1}{2}$ and $r\in \{5,7\}$.
\item[(2)]
When $\frac{q-1}{12}=2\tau+1$ in Theorem \ref{th 11} (2),
let $r=23$, then there exists a quantum MDS code with parameters $[[\frac{23}{24}(q^2-1)+1,\frac{23}{24}(q^2-1)-2k+1,k+1]]_q$
for any $1\leq k\leq \frac{23(q-1)}{24}+\frac{1}{2}$.
\end{itemize}
\end{example}

Applying the propagation rule for Theorem \ref{th 11},
we can get the following corollary.

\begin{corollary}\label{cor 1}(\cite{RefJ (2019) F.Tian})
Let $q$ be an odd integer and $n=(\frac{h}{2}+1)\frac{q^2-1}{2h}$, where $\frac{q-1}{h}=2\tau+1$ for some $\tau\geq 1$. Then for $1\leq k\leq (h+1)\frac{q-1}{2h}-\frac{1}{2}$,
there exists a quantum MDS code with parameters $[[n,n-2k,k+1]]_q$.
\end{corollary}

\begin{corollary}\label{cor 2}
Let $q$ be an odd integer and $n=\frac{r(q^2-1)}{2h}$, where $\frac{q-1}{h}=2\tau+1$ for some $\tau\geq 1$. Then for odd $h< r< h+1$ and $1\leq k\leq r\frac{q-1}{2h}-\frac{1}{2}$,
there exists a quantum MDS code with parameters $[[n,n-2k,k+1]]_q$.
\end{corollary}

\subsection{quantum MDS codes of length $n=r\frac{q^2-1}{2h}$, where $\frac{q-1}{h}$ is odd}

Throughout this section,
let $h$ be an even integer such that $\frac{q-1}{h}=2\tau+1$ for some $\tau \geq 1$.

\begin{lemma}\label{lem sec4 3}
Let $q$, $h$, $\tau$ and $m$ be defined as above.
Suppose $1\leq k\leq (h+2t)\frac{q-1}{2h}$, where $1\leq t\leq \frac{h}{2}$.
Then for $0\leq i,j\leq k-1$, $qi+j+\frac{q+1}{2}=sm$ if and only if
$s\in [2,4,\dots,2t-2]\cup[h+2,h+4,\dots,h+2t-2]\cup \{1,3,\dots,h+2t-1\}$.
\end{lemma}

\begin{proof}
Note that $0\leq i,j<q-1$,
we have $0\leq qi+j< q^2-1$.
If $qi+j+\frac{q+1}{2}=sm=s\frac{q^2-1}{2h}$,
then $1\leq s< 2h$.
Similarly, we divide the proof into the following two parts.

$\bullet$ $\textbf{Case 1:}$ $s$ is even.

For $2\leq s\leq h$,
we have
\[\begin{split} qi+j&=s\frac{q^2-1}{2h}-\frac{q+1}{2}\\
                           &=(\frac{s}{2}\frac{q-1}{h}-1)q+\frac{h+s}{2}\frac{q-1}{h}.\\
	\end{split}\]
It follows that
$$i=\frac{s}{2}\frac{q-1}{h}-1,\quad j=\frac{h+s}{2}\frac{q-1}{h}.$$

If $2t\leq s\leq h$, then
$$j=\frac{h+s}{2}\frac{q-1}{h}\geq (h+2t)\frac{q-1}{2h}\geq k,$$
which contradicts to the fact $j\leq k-1$;

On other hand, for $h+2\leq s\leq 2h$,
we have
\[\begin{split} qi+j&=s\frac{q^2-1}{2h}-\frac{q+1}{2}\\
                           &=(\frac{s}{2}\frac{q-1}{h})q+\frac{s-h}{2}\frac{q-1}{h}-1.\\
	\end{split}\]
It follows that
$$i=\frac{s}{2}\frac{q-1}{h},\quad j=\frac{s-h}{2}\frac{q-1}{h}-1.$$

If $h+2t\leq s\leq 2h$, then
$$i=\frac{s}{2}\frac{q-1}{h}\geq (h+2t)\frac{q-1}{2h}\geq k,$$
which contradicts to the fact $i\leq k-1$.
Hence, $s\in [2,4,\dots,2t-2]\cup[h+2,h+4,\dots,h+2t-2]$.

$\bullet$ $\textbf{Case 2:}$ $s$ is odd.

For $1\leq s\leq 2h-1$,
we have
\[\begin{split} qi+j&=s\frac{q^2-1}{2h}-\frac{q+1}{2}\\
                           &=(\frac{s-1}{2}\frac{q-1}{h}+\tau)q+\frac{s-1}{2}\frac{q-1}{h}+\tau.\\
	\end{split}\]
It follows that
$$i=j=\frac{s-1}{2}\frac{q-1}{h}+\tau.$$

If $h+2t+1\leq s\leq 2h-1$, then
$$i=\frac{s-1}{2}\frac{q-1}{h}+\tau\geq (h+2t)\frac{q-1}{2h}\geq k,$$
which contradicts to the fact $i\leq k-1$.
Hence, $s\in \{1,3,\dots,h+2t-1\}$.
This completes the proof.
\end{proof}

\begin{theorem}\label{th 22}
Let $q$ be an odd integer and $n=\frac{r(q^2-1)}{2h}$, where $\frac{q-1}{h}=2\tau+1$ for some $\tau\geq 1$. Then for $\frac{h}{2}+1< r\leq h$ and $1\leq k\leq (h+2)\frac{q-1}{2h}$,
there exists a quantum MDS code with parameters $[[n,n-2k,k+1]]_q$.
\end{theorem}

\begin{proof}
Set
$$\bm{a}=(\theta^{i_1}(\langle \gamma\rangle),\theta^{i_2}(\langle \gamma\rangle),\dots,\theta^{i_r}(\langle \gamma\rangle))\in \F_{q^2}^{r\frac{q^2-1}{2h}},$$
where $i_1,i_2,\dots,i_r$ are distinct modulo $2h$.
Let
$$\bm{v}=(v_1,v_1\theta^{h},\dots,v_1\theta^{(m-1)h},\dots,v_{r},v_r\theta^{h}\dots,v_{r}\theta^{(m-1)h})\in (\F_{q^2}^*)^{r\frac{q^2-1}{2h}},$$
where $v_1,v_2\dots,v_r\in \F_{q^2}^*$.
Then for any $0\leq i,j< q-1$, we have
$$\langle \bm{a}^{qi+j},\bm{v}^{q+1}\rangle=\sum_{l=1}^r\theta^{i_l(qi+j)}v_l^{q+1}\sum_{\nu=0}^{m-1}\gamma^{\nu(qi+j+\frac{q+1}{2})},$$
thus
\[\langle \bm{a}^{qi+j},\bm{v}^{q+1}\rangle=\begin{cases}
0, & if\ m\nmid (qi+j+\frac{q+1}{2}); \\
m\sum_{l=1}^r\theta^{i_l(qi+j)}v_l^{q+1}, & if\  m\mid (qi+j+\frac{q+1}{2}).
\end{cases}\]

Since $1\leq k\leq (h+2)\frac{q-1}{2h}$, where $t=1$,
by Lemma \ref{lem sec4 3},
when $m\mid (qi+j+\frac{q+1}{2})$, we have
$\langle \bm{a}^{qi+j},\bm{v}^{q+1}\rangle=m\sum_{l=1}^r\alpha^{i_ls}\xi^{i_l}v_l^{q+1},$
where $s\in \{1,3,\dots,h+1\}$.
Since $r\leq h$, we can assume that $i_1,i_2,\dots,i_r$ are distinct modulo $h$.
Let
$$A=\begin{pmatrix}
  \alpha^{i_1} \xi^{i_1} &  \alpha^{i_2} \xi^{i_2} & \dots & \alpha^{i_r}\xi^{i_r}\\
  \alpha^{i_1}\beta^{i_1} \xi^{i_1} &  \alpha^{i_2}\beta^{i_2} \xi^{i_2} & \dots & \alpha^{i_r}\beta^{i_r} \xi^{i_r}\\
  \vdots & \vdots &  \ddots & \vdots\\
  \alpha^{i_1}\beta^{\frac{h}{2}i_1} \xi^{i_1} &  \alpha^{i_2}\beta^{\frac{h}{2}i_2} \xi^{i_2} & \dots & \alpha^{i_r}\beta^{\frac{h}{2}i_r} \xi^{i_r}\\
\end{pmatrix}$$
be an $(\frac{h}{2}+1)\times r$ matrix over $\F_{q^2}$.
Note that $\beta^q=\beta$, $\alpha^q=-\alpha$ and $\xi^q=-\xi$,
it follows that $\beta\in \F_q$ and $\alpha\xi\in \F_q$.
Therefore, $A$ is a matrix over $\F_q$.
Similarly, we can obtain that the equation $A\bm{u}^T=\bm{0}^{T}$ has a solution $\bm{u}=(u_1,u_2,\dots,u_r)\in (\F_{q}^*)^{r}$.
Let $v_l^{q+1}=u_l$ for $1\leq l\leq r$.
When $m\mid (qi+j+\frac{q+1}{2})$,
we have
$$\langle \bm{a}^{qi+j},\bm{v}^{q+1}\rangle=m\sum_{l=1}^r\theta^{i_lsm}\xi^{i_l}v_l^{q+1}
=m\sum_{l=1}^r\alpha^{i_ls}\xi^{i_l}u_l=0.$$
Then $\langle \bm{a}^{qi+j},\bm{v}^{q+1}\rangle=0$, for all $0\leq i,j\leq k-1$.
By Lemmas \ref{lem GRS panding} and \ref{Lem zhuyao}, we have the conclusion.
This completes the proof.
\end{proof}

Taking $r=\frac{h}{2}+2$ in Theorem \ref{th 22}, we obtain the following known results.

\begin{corollary}\label{cor 3}(\cite{RefJ (2019) F.Tian})
Let $q$ be an odd integer and $n=(\frac{h}{2}+2)\frac{q^2-1}{2h}$, where $\frac{q-1}{h}=2\tau+1$ for some $\tau\geq 1$. Then for any $1\leq k\leq (h+2)\frac{q-1}{2h}$,
there exists a quantum MDS code with parameters $[[n,n-2k,k+1]]_q$.
\end{corollary}

\begin{example}
By Theorem \ref{th 22} and Corollary \ref{cor 2}, we obtain the following examples of quantum MDS codes.
\begin{itemize}
\item[(1)]
When $\frac{q-1}{6}=2\tau+1$ in Theorem \ref{th 22},
let $r=5$,
then there exists a quantum MDS code with parameters
 $[[\frac{5}{12}(q^2-1),\frac{5}{12}(q^2-1)-2k,k+1]]_q$
for any $1\leq k\leq \frac{2(q-1)}{3}$.
\item[(2)]
When $\frac{q-1}{10}=2\tau+1$ in Corollary \ref{cor 2},
let $r=19$, then there exists a quantum MDS code with parameters $[[\frac{19}{20}(q^2-1),\frac{19}{20}(q^2-1)-2k,k+1]]_q$
for any $1\leq k\leq \frac{19(q-1)}{20}-\frac{1}{2}$.
\end{itemize}
\end{example}

\subsection{quantum MDS codes of length $n=r\frac{q^2-1}{2h}+1$, where $\frac{q+1}{h}$ is odd}

Throughout this section,
let $h$ be an even integer such that $\frac{q+1}{h}=2\tau+1$ for some $\tau \geq 1$.

\begin{lemma}\label{lem sec 4 4}
Let $q$, $h$, $\tau$ and $m$ be defined as above. The following statements hold.
\begin{itemize}
\item[(1)] Suppose $1\leq k\leq (h+2t+1)\frac{q+1}{2h}-\frac{3}{2}$, where $0\leq t\leq \frac{h}{2}-1$. Then for $0\leq i,j\leq k-1$, $qi+j=sm$ if and only if
$s\in \{0\}\cup\{1,3,\dots,h-1\}\cup \{h-2t, h-2t+2,\dots, h+2t\}\cup [2h-2t+1,2h-2t+3,\dots,h+2t-1]$.
\item[(2)] Suppose $1\leq k\leq (h+2t)\frac{q+1}{2h}-1$, where $1\leq t\leq \frac{h}{2}$. Then for $0\leq i,j\leq k-1$, $qi+j=sm$ if and only if
$s\in \{0\}\cup\{1,3,\dots,h-1\}\cup \{h-2t+2, h-2t+4,\dots, h+2t-2\}\cup [2h-2t+1,2h-2t+3,\dots,h+2t-1]$.
\end{itemize}
\end{lemma}

\begin{proof}
Since the proof of (2) is completely similar, we only prove (1).
Similarly to Lemma \ref{lem sec 4 1}, we can obtain $0\leq s< 2h$ and divide the proof into the following two parts.

$\bullet$ $\textbf{Case 1:}$ $s$ is odd.

For $h+1\leq s\leq 2h-1$,
we have
\[\begin{split} qi+j &=(\frac{s-1}{2}\frac{q+1}{h})q+\frac{q-1}{2}\frac{q+1}{h}-\frac{s-1}{2}\frac{q+1}{h}\\
                     &=(\frac{s-1}{2}\frac{q+1}{h}+\tau-1)q+q+\frac{h-s+1}{2}\frac{q+1}{h}-\tau-1\\
                     &=(\frac{s}{2}\frac{q+1}{h}-\frac{3}{2})q+\frac{3h-s}{2}\frac{q+1}{h}-\frac{3}{2}.\\
	\end{split}\]
It follows that
$$i=\frac{s}{2}\frac{q+1}{h}-\frac{3}{2},\quad j=\frac{3h-s}{2}\frac{q+1}{h}-\frac{3}{2}.$$

If $s\geq h+2t+1$, then
$$i=\frac{s}{2}\frac{q+1}{h}-\frac{3}{2}\geq (h+2t+1)\frac{q+1}{2h}-\frac{3}{2}\geq k,$$
which contradicts to the fact $i\leq k-1$;

If $s\leq 2h-2t-1$, then
$$j=\frac{3h-s}{2}\frac{q+1}{h}-\frac{3}{2}\geq  (h+2t+1)\frac{q+1}{2h}-\frac{3}{2}\geq k,$$
which contradicts to the fact $j\leq k-1$.
Hence, we have $s\in \{1,3,\dots,h-1\}\cup [2h-2t+1,2h-2t+3,\dots,h+2t-1]$.

$\bullet$ $\textbf{Case 2:}$ $s$ is even.

For $2\leq s\leq 2h-2$, we have
\[\begin{split} qi+j &=\frac{s}{2}\frac{q+1}{h}q-\frac{s}{2}\frac{q+1}{h}\\
                     &=(\frac{s}{2}\frac{q+1}{h}-1)q+\frac{2h-s}{2}\frac{q+1}{h}-1.\\
	\end{split}\]
It follows that
$$i=\frac{s}{2}\frac{q+1}{h}-1,\quad j=\frac{2h-s}{2}\frac{q+1}{h}-1.$$

If $s\leq h-2t-2$, then
$$ j=\frac{2h-s}{2}\frac{q+1}{h}-1 \geq (h+2t+2)\frac{q+1}{2h}-1\geq k,$$
which contradicts to the fact $j\leq k-1$;

If $s\geq h+2t+2$, then
$$i=\frac{s}{2}\frac{q+1}{h}-1\geq (h+2t+2)\frac{q+1}{2h}-1\geq k,$$
which contradicts to the fact $i\leq k-1$.
Hence, we have $s\in \{h-2t, h-2t+4,\dots, h+2t\}$.
This completes the proof.
\end{proof}

\begin{theorem}\label{th 33}
Let $q$ be an odd integer and $n=\frac{r(q^2-1)}{2h}+1$, where $\frac{q+1}{h}=2\tau+1$ for some $\tau\geq 1$.
The following statements hold.
\begin{itemize}
\item[(1)] For odd $h< r<\frac{3h}{2}$ and $1\leq k\leq (r+1)\frac{q+1}{2h}-1$,
  there exists a quantum MDS code with parameters $[[n,n-2k,k+1]]_q$.
\item[(2)] For odd $\frac{3h}{2}<r<2h$ and $1\leq k\leq r\frac{q+1}{2h}-\frac{3}{2}$,
  there exists a quantum MDS code with parameters $[[n,n-2k,k+1]]_q$.
\end{itemize}
\end{theorem}

\begin{proof}
(1) Define $\bm{a}$, $\bm{v}$ as in Theorem \ref{th 11}. Since $1\leq k\leq (h+2t)\frac{q+1}{2h}-1$,
where $t=\frac{r-h+1}{2}$,
by Lemma \ref{lem sec 4 4} (2),
when $m\mid (qi+j)$, $(i,j)\neq (0,0)$, we have
$\langle \bm{a}^{qi+j},\bm{v}^{q+1}\rangle =m\sum_{l=1}^r\alpha^{i_ls}v_l^{q+1},$
where $s\in \{1,3,\dots,h-1\}\cup \{h-2t+2, h-2t+4,\dots, h+2t-2\}$. Let
$$B=\begin{pmatrix}
  \alpha^{i_1} &  \alpha^{i_2} & \dots & \alpha^{i_r}\\
  \alpha^{3i_1} &  \alpha^{3i_2} & \dots & \alpha^{3i_r}\\
  \vdots & \vdots &  \ddots & \vdots\\
  \alpha^{(h-1)i_1} &  \alpha^{(h-1)i_2} & \dots & \alpha^{(h-1)i_r}\\
\end{pmatrix},$$
$$C_1=\begin{pmatrix}
  \alpha^{(h-2t+2)i_1} &  \alpha^{(h-2t+2)i_2} & \dots & \alpha^{(h-2t+2)i_r}\\
  \alpha^{(h-2t+4)i_1} &  \alpha^{(h-2t+4)i_2} & \dots & \alpha^{(h-2t+4)i_r}\\
  \vdots & \vdots &  \ddots & \vdots\\
  \alpha^{(h+2t-2)i_1} &  \alpha^{(h+2t-2)i_2} & \dots & \alpha^{(h+2t-2)i_r}\\
\end{pmatrix}
\quad and\quad A=\begin{pmatrix}
1& \bm{1}_r\\
\bm{0}^T & B\\
\bm{0}^T & C_1\\
\end{pmatrix}
$$
be the matrices over $\F_{q^2}$, respectively.
Let
$$\tilde{A}=\begin{pmatrix}
1& 1&  1&  \dots &1\\
0&  \alpha^{i_1} &  \alpha^{i_2} & \dots & \alpha^{i_r}\\
 \vdots &   \vdots & \vdots &  \ddots & \vdots\\
 0& \alpha^{(h+2t-2)i_1} &  \alpha^{(h+2t-2)i_2} & \dots & \alpha^{(h+2t-2)i_r}\\
\end{pmatrix}$$
be an $r\times (r+1)$ matrix over $\F_{q^2}$.
Similar to the proof of Theorem \ref{th 11}, we can obtain that
the equation $A\bm{u}^T=\bm{0}^{T}$ has a solution $\bm{u}=(u_0,u_1,\dots,u_r)\in (\F_{q^2}^*)^{r+1}$.
Note that $\alpha^h=-1$, $\alpha^q=-\alpha^{-1}$ and $\alpha^{2h}=1$, then for any odd number $1\leq s\leq h-1$, and for any $1\leq j\leq r$, we have
\[\begin{split} \alpha^{si_jq}&=(-\alpha^{-s})^{i_j}\\
                           &=\alpha^{(h-s)i_j}.\\
	\end{split}\]
Then, $B$ is row equivalent to $B^{(q)}$.
For any even number $-2t\leq s\leq 2t$, and for any $1\leq j\leq r$, we have
\[\begin{split}        \alpha^{(h+s)i_jq}&=(\alpha^{-h-s})^{i_j}\\
                           &=\alpha^{(h-s)i_j}.\\
	\end{split}\]
Then, $C_1$ is row equivalent to $C_1^{(q)}$.
It follows that $A$ is row equivalent to $A^{(q)}$.
By Lemma \ref{lem two youjie 2}, the equation $A\bm{u}^T=\bm{0}^{T}$ has a solution $\bm{u}=(u_0,u_1,\dots,u_r)\in (\F_{q}^*)^{r+1}$.

(2)
Since $1\leq k\leq (h+2t+1)\frac{q+1}{2h}-\frac{3}{2}$,
where $t=\frac{r-h-1}{2}$,
by Lemma \ref{lem sec 4 4} (1),
when $m\mid (qi+j)$, $(i,j)\neq (0,0)$, we have
$\langle \bm{a}^{qi+j},\bm{v}^{q+1}\rangle =m\sum_{l=1}^r\alpha^{i_ls}v_l^{q+1},$
where $s\in \{1,3,\dots,h-1\}\cup \{h-2t, h-2t+4,\dots, h+2t\}\cup \{2h-2t+1,2h-2t+3,\dots,h+2t-1\}$.
Let
$$D=\begin{pmatrix}
  \alpha^{(2h-2t+1)i_1} &  \alpha^{(2h-2t+1)i_2} & \dots & \alpha^{(2h-2t+1)i_r}\\
  \alpha^{(2h-2t+3)i_1} &  \alpha^{(2h-2t+3)i_2} & \dots & \alpha^{(2h-2t+3)i_r}\\
  \vdots & \vdots &  \ddots & \vdots\\
  \alpha^{(h+2t-1)i_1} &  \alpha^{(h+2t-1)i_2} & \dots & \alpha^{(h+2t-1)i_r}\\
\end{pmatrix}$$
$$C_2=\begin{pmatrix}
  \alpha^{(h-2t)i_1} &  \alpha^{(h-2t)i_2} & \dots & \alpha^{(h-2t)i_r}\\
  \alpha^{(h-2t+2)i_1} &  \alpha^{(h-2t+2)i_2} & \dots & \alpha^{(h-2t+2)i_r}\\
  \vdots & \vdots &  \ddots & \vdots\\
  \alpha^{(h+2t)i_1} &  \alpha^{(h+2t)i_2} & \dots & \alpha^{(h+2t)i_r}\\
\end{pmatrix}
\quad and \quad
A=\begin{pmatrix}
1& \bm{1}_r\\
\bm{0}^T&  B\\
\bm{0}^T& C_2\\
\bm{0}^T& D\\
\end{pmatrix}
$$
be the matrices over $\F_{q^2}$, respectively.
Let
$$
\tilde{A}=\begin{pmatrix}
1& 1&  1&  \dots &1\\
0&  \alpha^{i_1} &  \alpha^{i_2} & \dots & \alpha^{i_r}\\
 \vdots &   \vdots & \vdots &  \ddots & \vdots\\
 0& \alpha^{(h+2t)i_1} &  \alpha^{(h+2t)i_2} & \dots & \alpha^{(h+2t)i_r}\\
\end{pmatrix}$$
be an $r\times (r+1)$ matrix over $\F_{q^2}$.
Similarly, we can obtain that
the equation $A\bm{u}^T=\bm{0}^{T}$ has a solution $\bm{u}=(u_0,u_1,\dots,u_r)\in (\F_{q^2}^*)^{r+1}$ and $C_2$ is row equivalent to ${C_2}^{(q)}$.
For any number $s\in \{\frac{h}{2}-2t+1,\frac{h}{2}-2t+3,\dots,-\frac{h}{2}+2t-1\}$, and for any $1\leq j\leq r$, we have
\[\begin{split}        \alpha^{(\frac{3h}{2}+s)i_jq}&=(-\alpha^{-\frac{3h}{2}-s})^{i_j}\\
                           &=(-\alpha^{\frac{h}{2}-s})^{i_j}\\
                           &=(\alpha^{\frac{3h}{2}-s})^{i_j}.\\
	\end{split}\]
Hence, $D$ is row equivalent to ${D}^{(q)}$.
It follows that $A$ is row equivalent to $A^{(q)}$.
By Lemma \ref{lem two youjie 2}, the equation $A\bm{u}^T=\bm{0}^{T}$ has a solution $\bm{u}=(u_0,u_1,\dots,u_r)\in (\F_{q}^*)^{r+1}$.

The following proof is similar to Theorem \ref{th 11}.
We omit the details. This completes the proof.
\end{proof}

\begin{example}
By taking $h=14,18$ in Theorem \ref{th 33}, we get the following quantum MDS codes.
\begin{itemize}
\item[(1)]
When $\frac{q+1}{14}=2\tau+1$ in Theorem \ref{th 33} (1),
let $r=19$,
then there exists a quantum MDS code with parameters
 $[[\frac{19}{28}(q^2-1)+1,\frac{19}{28}(q^2-1)-2k+1,k+1]]_q$
for any $1\leq k\leq \frac{5(q+1)}{7}-1$.
\item[(2)]
When $\frac{q+1}{18}=2\tau+1$ in Theorem \ref{th 33} (2),
let $r=35$, then there exists a quantum MDS code with parameters $[[\frac{35}{36}(q^2-1)+1,\frac{35}{36}(q^2-1)-2k+1,k+1]]_q$
for any $1\leq k\leq \frac{35(q+1)}{36}-\frac{3}{2}$.
\end{itemize}
\end{example}

\subsection{quantum MDS codes of length $n=r\frac{q^2-1}{2h}$, where $\frac{q+1}{h}$ is odd}

Throughout this section,
let $h$ be an even integer such that $\frac{q+1}{h}=2\tau+1$ for some $\tau \geq 1$.

\begin{lemma}\label{lem sec 4 5}
Let $q$, $h$, $\tau$ and $m$ be defined as above. The following statements hold.
\begin{itemize}
\item[(1)] Suppose $1\leq k\leq (h+2t+1)\frac{q+1}{2h}-\frac{3}{2}$, where $0\leq t<\frac{h}{2}$. Then for any $0\leq i,j\leq k-1$, $qi+j+\frac{q+1}{2}=sm$ if and only if
$s\in \{2,4,\dots,h-2\}\cup [h-2t+1,h-2t+3,\dots,h+2t-1]\cup[2h-2t,2h-2t+2,\dots,h+2t]$.
\item[(2)] Suppose $1\leq k\leq (h+2t)\frac{q+1}{2h}-2$, where $1\leq t\leq \frac{h}{2}$. Then for any $0\leq i,j\leq k-1$, $qi+j+\frac{q+1}{2}=sm$ if and only if
$s\in \{2,4,\dots,h-2\}\cup [h-2t+1,h-2t+3,\dots,h+2t-1]\cup [2h-2t+2,2h-2t+4,\dots,h+2t-2]$.
\end{itemize}
\end{lemma}

\begin{proof}
Since the proof of (2) is completely similar, we only prove (1).
Similarly to Lemma \ref{lem sec4 3}, we can obtain $1\leq s< 2h$ and divide the proof into the following two parts.

$\bullet$ $\textbf{Case 1:}$ $s$ is odd.

For $1\leq s\leq 2h-1$,
we have
\[\begin{split} qi+j &=(\frac{s-1}{2}\frac{q+1}{h}+\tau)q-\frac{s-1}{2}\frac{q+1}{h}-\tau-1\\
                     &=(\frac{s-1}{2}\frac{q+1}{h}+\tau-1)q+q+1-\frac{s-1}{2}\frac{q+1}{h}-\tau-2\\
                     &=(\frac{s}{2}\frac{q+1}{h}-\frac{3}{2})q+\frac{2h-s}{2}\frac{q+1}{h}-\frac{3}{2}.\\
	\end{split}\]
It follows that
$$i=\frac{s}{2}\frac{q+1}{h}-\frac{3}{2},\quad j=\frac{2h-s}{2}\frac{q+1}{h}-\frac{3}{2}.$$

If $s\geq h+2t+1$, then
$$i=\frac{s}{2}\frac{q+1}{h}-\frac{3}{2}\geq (h+2t+1)\frac{q+1}{2h}-\frac{3}{2}\geq k,$$
which contradicts to the fact $i\leq k-1$;

If $s\leq h-2t-1$, then
$$j=\frac{2h-s}{2}\frac{q+1}{h}-\frac{3}{2}\geq  (h+2t+1)\frac{q+1}{2h}-\frac{3}{2}\geq k,$$
which contradicts to the fact $j\leq k-1$.
Hence, we have $s\in [h-2t+1,h-2t+3,\dots,h+2t-1]$.

$\bullet$ $\textbf{Case 2:}$ $s$ is even.

For $h\leq s\leq 2h$, we have
\[\begin{split} qi+j &=\frac{s}{2}\frac{q+1}{h}q-\frac{s}{2}\frac{q+1}{h}-\frac{q+1}{2}\\
                     &=(\frac{s}{2}\frac{q+1}{h}-2)q+\frac{3h-s}{2}\frac{q+1}{h}-2.\\
	\end{split}\]
It follows that
$$i=\frac{s}{2}\frac{q+1}{h}-2,\quad j=\frac{3h-s}{2}\frac{q+1}{h}-2.$$

If $s\leq 2h-2t-2$, then
$$ j=\frac{3h-s}{2}\frac{q+1}{h}-2 \geq (h+2t+2)\frac{q+1}{2h}-2\geq k,$$
which contradicts to the fact $j\leq k-1$;

If $s\geq h+2t+2$, then
$$i=\frac{s(q+1)}{2h}-2\geq \frac{(h+2t+2)(q+1)}{2h}-2\geq k,$$
which contradicts to the fact $i\leq k-1$.
Hence, we have $s\in \{2,4,\dots,h-2\}\cup [2h-2t,2h-2t+2,\dots,h+2t]$.
This completes the proof.
\end{proof}

\begin{theorem}\label{th 44}
Let $q$ be an odd integer and $n=\frac{r(q^2-1)}{2h}$, where $\frac{q+1}{h}=2\tau+1$ for some $\tau\geq 1$.
The following statements hold.
\begin{itemize}
\item[(1)] For $\frac{h}{2}\leq r\leq h$ and $1\leq k\leq (h+1)\frac{q+1}{2h}-\frac{3}{2}$,
  there exists a quantum MDS code with parameters $[[n,n-2k,k+1]]_q$.
\item[(2)] For odd $h< r<\frac{3h}{2}$ and $1\leq k\leq (r+2)\frac{q+1}{2h}-\frac{3}{2}$,
  there exists a quantum MDS code with parameters $[[n,n-2k,k+1]]_q$.
\item[(3)] For odd $\frac{3h}{2}\leq r< 2h$ and $1\leq k\leq (r+1)\frac{q+1}{2h}-2$,
  there exists a quantum MDS code with parameters $[[n,n-2k,k+1]]_q$.
\end{itemize}
\end{theorem}

\begin{proof}
(1)
Define $\bm{a}$, $\bm{v}$ as in Theorem \ref{th 22}.
Since $1\leq k\leq (h+1)\frac{q+1}{2h}-\frac{3}{2}$, where $t=0$,
by Lemma \ref{lem sec 4 5} (1),
when $m\mid (qi+j+\frac{q+1}{2})$, we have
$\langle \bm{a}^{qi+j},\bm{v}^{q+1}\rangle =m\sum_{l=1}^r\alpha^{i_ls}\xi^{i_l}v_l^{q+1},$
where $s\in \{2,4,\dots,h-2\}$.
Since $r\leq h$, we can assume that $i_1,i_2,\dots,i_r$ are distinct modulo $h$.
Let
$$B=\begin{pmatrix}
  \beta^{i_1} \xi^{i_1} &  \beta^{i_2} \xi^{i_2} & \dots & \beta^{i_r}\xi^{i_r}\\
  \beta^{2i_1} \xi^{i_1} &  \beta^{2i_2} \xi^{i_2} & \dots & \beta^{2i_r}\xi^{i_r}\\
  \vdots & \vdots &  \ddots & \vdots\\
  \beta^{(\frac{h}{2}-1)i_1} \xi^{i_1} &  \beta^{(\frac{h}{2}-1)i_2} \xi^{i_2} & \dots & \beta^{(\frac{h}{2}-1)i_r}\xi^{i_r}\\
\end{pmatrix}$$
be an $(\frac{h}{2}-1)\times r$ matrix over $\F_{q^2}$.
Let $B_i\ (1\leq i\leq r)$ be the $(\frac{h}{2}-1)\times (r-1)$ matrix obtained from $B$ by deleting the $i$-th column.
We can get $rank(B)=rank(B_1)=\dots=rank(B_r)=\frac{h}{2}-1$.
Note that $\beta^q=\beta^{-1}$, $\beta^{\frac{h}{2}}=-1$ and $\xi^q=-\xi$,
then for any $1\leq s\leq \frac{h}{2}-1$, and for any $1\leq j\leq r$, we have
\[\begin{split}        (\beta^{si_j}\xi^{i_j})^q&=(-\beta^{-s}\xi)^{i_j}\\
                           &=\beta^{(\frac{h}{2}-s)i_j}\xi^{i_j}.\\
	\end{split}\]
It follows that $B$ is row equivalent to $B^{(q)}$.
By Lemma \ref{lem two youjie 1} and Lemma \ref{lem two youjie 2}, the equation $B\bm{u}^T=\bm{0}^{T}$ has a solution $\bm{u}=(u_1,u_2,\dots,u_r)\in (\F_{q}^*)^{r}$.

(2)
Since $1\leq k\leq (h+2t+1)\frac{q+1}{2h}-\frac{3}{2}$, where $t=\frac{r-h+1}{2}$,
by Lemma \ref{lem sec 4 5} (1),
when $m\mid (qi+j+\frac{q+1}{2})$, we have
$\langle \bm{a}^{qi+j},\bm{v}^{q+1}\rangle=m\sum_{l=1}^r\alpha^{i_ls}\xi^{i_l}v_l^{q+1},$
where $s\in \{2,4,\dots,h-2\}\cup \{h-2t+1,h-2t+3,\dots,h+2t-1\}$.
Let
$$C=\begin{pmatrix}
  \alpha^{(h-2t+1)i_1} \xi^{i_1} &  \alpha^{(h-2t+1)i_2} \xi^{i_2} & \dots & \alpha^{(h-2t+1)i_r} \xi^{i_r}\\
  \alpha^{(h-2t+3)i_1} \xi^{i_1} &  \alpha^{(h-2t+3)i_2} \xi^{i_2} & \dots & \alpha^{(h-2t+3)i_r}\xi^{i_r}\\
  \vdots & \vdots &  \ddots & \vdots\\
  \alpha^{(h+2t-1)i_1} \xi^{i_1} &  \alpha^{(h+2t-1)i_2} \xi^{i_2} & \dots & \alpha^{(h+2t-1)i_r}\xi^{i_r}\\
\end{pmatrix}
\quad and\quad A=\begin{pmatrix}
 B\\
 C\\
\end{pmatrix}
$$
be the matrices over $\F_{q^2}$, respectively.
Let
$$\tilde{A}=\begin{pmatrix}
\alpha^{2i_1} \xi^{i_1} &  \alpha^{2i_2}\xi^{i_2} & \dots & \alpha^{2i_r}\xi^{i_r}\\
\alpha^{3i_1} \xi^{i_1} &  \alpha^{3i_2} \xi^{i_2} & \dots & \alpha^{3i_r}\xi^{i_r}\\
 \vdots & \vdots &  \ddots & \vdots\\
\alpha^{(h+2t-1)i_1} \xi^{i_1} &  \alpha^{(h+2t-1)i_2}\xi^{i_2} & \dots & \alpha^{(h+2t-1)i_r}\xi^{i_r}\\
\end{pmatrix}$$
be an $(r-1)\times r$ matrix over $\F_{q^2}$.
Similar to the proof of Theorem \ref{th 11}, we can obtain that the equation $A\bm{u}^T=\bm{0}^{T}$ has a solution $\bm{u}=(u_1,u_2,\dots,u_r)\in (\F_{q^2}^*)^{r}$.
Note that $\alpha^h=-1$, $\alpha^q=-\alpha^{-1}$ and $\alpha^{2h}=1$, then for any odd number $-2t+1\leq s\leq 2t-1$, and for any $1\leq j\leq r$, we have
\[\begin{split}        (\alpha^{(h+s)i_j}\xi^{i_j})^q&=(\alpha^{-h-s}\xi)^{i_j}\\
                           &=\alpha^{(h-s)i_j}\xi^{i_j}.\\
	\end{split}\]
Then, $C$ is row equivalent to $C^{(q)}$.
It follows that $A$ is row equivalent to $A^{(q)}$.
By Lemma \ref{lem two youjie 2}, the equation $A\bm{u}^T=\bm{0}^{T}$ has a solution $\bm{u}=(u_1,u_2,\dots,u_r)\in (\F_{q}^*)^{r}$.

(3)
Since $1\leq k\leq (h+2t)\frac{q+1}{2h}-2$,
where $t=\frac{r-h+1}{2}$,
by Lemma \ref{lem sec 4 5} (2),
when $m\mid (qi+j+\frac{q+1}{2})$, we have
$\langle \bm{a}^{qi+j},\bm{v}^{q+1}\rangle=m\sum_{l=1}^r\alpha^{i_ls}\xi^{i_l}v_l^{q+1},$
where $s\in \{2,4,\dots,h-2\}\cup \{h-2t+1,h-2t+3,\dots,h+2t-1\}\cup \{2h-2t+2,2h-2t+4,\dots,h+2t-2\}$.
Let
$$D=\begin{pmatrix}
  \alpha^{(2h-2t+2)i_1} \xi^{i_1} &  \alpha^{(2h-2t+2)i_2}\xi^{i_2}  & \dots & \alpha^{(2h-2t+2)i_r}\xi^{i_r} \\
  \alpha^{(2h-2t+4)i_1}\xi^{i_1}  &  \alpha^{(2h-2t+4)i_2} \xi^{i_2} & \dots & \alpha^{(2h-2t+4)i_r}\xi^{i_r} \\
  \vdots & \vdots &  \ddots & \vdots\\
  \alpha^{(h+2t-2)i_1}\xi^{i_1}  &  \alpha^{(h+2t-2)i_2}\xi^{i_2}  & \dots & \alpha^{(h+2t-2)i_r}\xi^{i_r} \\
\end{pmatrix}
\quad and \quad
A=\begin{pmatrix}
 B\\
 C\\
D\\
\end{pmatrix}
$$
be the matrices over $\F_{q^2}$, respectively.
Let $\tilde{A}$ be defined as in (2).
Similarly, we can obtain that
the equation $A\bm{u}^T=\bm{0}^{T}$ has a solution $\bm{u}=(u_1,u_2,\dots,u_r)\in (\F_{q^2}^*)^{r}$.
For any number $s\in \{\frac{h}{2}-2t+2,\frac{h}{2}-2t+4,\dots,-\frac{h}{2}+2t-2\}$, and for any $1\leq j\leq r$, we have
\[\begin{split}        (\alpha^{(\frac{3h}{2}+s)i_j}\xi^{i_j})^q&=(-\alpha^{-\frac{3h}{2}-s}\xi)^{i_j}\\
                           &=(-\alpha^{\frac{h}{2}-s}\xi)^{i_j}\\
                           &=\alpha^{(\frac{3h}{2}-s)i_j}\xi^{i_j}\\
	\end{split}\]
Hence, $D$ is row equivalent to ${D}^{(q)}$.
It follows that $A$ is row equivalent to $A^{(q)}$.
By Lemma \ref{lem two youjie 2}, the equation $A\bm{u}^T=\bm{0}^{T}$ has a solution $\bm{u}=(u_1,u_2,\dots,u_r)\in (\F_{q}^*)^{r}$.

The following proof is similar to Theorem \ref{th 22}.
We omit the details. This completes the proof.
\end{proof}

In Theorem \ref{th 44} (3), take $r=2h-1$, we get the following result.

\begin{corollary}\label{cor 4}
Let $\frac{q+1}{h}=2\tau+1$ where $\tau\geq 1$.
Then for any $1\leq k\leq q-1$,
there exists a quantum MDS code with parameters
$[[\frac{(2h-1)(q^2-1)}{2h},\frac{(2h-1)(q^2-1)}{2h}-2k,k+1]]_q$.
\end{corollary}

\begin{example}
By taking $h=4,8,12$ in Theorem \ref{th 44}, we get the following quantum MDS codes.
\begin{itemize}
\item[(1)]
When $\frac{q+1}{4}=2\tau+1$ in Theorem \ref{th 44} (1),
let $r=3$,
then there exists a quantum MDS code with parameters
 $[[\frac{3}{8}(q^2-1),\frac{3}{8}(q^2-1)-2k,k+1]]_q$
for any $1\leq k\leq \frac{5(q+1)}{8}-\frac{3}{2}$.
\item[(2)]
When $\frac{q+1}{8}=2\tau+1$ in Theorem \ref{th 44} (2),
let $r=11$, then there exists a quantum MDS code with parameters $[[\frac{11}{16}(q^2-1),\frac{11}{16}(q^2-1)-2k,k+1]]_q$
for any $1\leq k\leq \frac{13(q+1)}{16}-\frac{3}{2}$.
\item[(3)]
When $\frac{q+1}{12}=2\tau+1$ in Theorem \ref{th 44} (3),
let $r=23$, then there exists a quantum MDS code with parameters $[[\frac{23}{24}(q^2-1),\frac{23}{24}(q^2-1)-2k,k+1]]_q$
for any $1\leq k\leq q-1$.
\end{itemize}
\end{example}

\section{Comparison}\label{sec5}

In this section, we compare our results with previous results to show that the codes constructed in this paper are new and have large minimum distance.
Firstly, we classify the lengths of known quantum MDS codes and summarize the results only for the largest known minimum distance.
We summarize some known results of quantum MDS codes in Table \ref{tab:1}.
With the above classification of lengths, we also summarize the quantum MDS codes obtained in this paper in Table \ref{tab:3}.

From Table \ref{tab:1}, it is easy to see that the lengths of most known quantum MDS codes satisfy $n\equiv 0,1($mod$\,q\pm1)$, except for some special lengths.
However, for the quantum MDS codes constructed in this paper, most of the lengths satisfy $n\equiv 0,1($mod$\,\frac{q\pm1}{2})$.
So when $r$ is odd, we construct quantum MDS codes with different lengths from the known results.
Moreover, we can find that the minimum distance of the quantum MDS codes constructed in this paper is close to $q$ when $r$ is large enough. More specifically, for some lengths the minimum distance can reach $q$, which is the largest minimum distance that can be constructed by GRS for that length (see Corollary \ref{cor 4}).
Hence, the quantum MDS codes constructed in this paper have large minimum distance.
In addition, we can see that for the cases with the same length, they can also be considered as a special case of our constructed codes. And the minimum distance of some known quantum MDS codes has been improved (see Corollaries \ref{cor 1}, \ref{cor 3} and Theorem \ref{th 11} (1)).

It is worth noting that applying the propagation rules to the codes in Table \ref{tab:1} also fails to yield the quantum MDS codes constructed in this paper.
Since most of the lengths of the quantum MDS codes in Table \ref{tab:1} satisfy $n\equiv 0,1($mod$\,q\pm1)$, the propagation rule can be used to obtain quantum codes of length $n\equiv 0,1($mod$\,\frac{q\pm1}{2})$. However, we can known that the minimum distance $d\leq q/2$ of these codes. Thus the quantum MDS codes constructed in this paper have new parameters.
Finally, by the propagation rules we can also obtain some new quantum MDS codes by the codes constructed in this paper.

\newcommand{\tabincell}[2]{\begin{tabular}{@{}#1@{}}#2\end{tabular}}
\begin{table}
\caption{Some known results of $[[n,n-2d+2,d]]_q$-quantum MDS codes with distances bigger than $q/2+1$}
\label{tab:1}
\begin{center}
\resizebox{\textwidth}{88mm}{
	\begin{tabular}{c|ccc}
		\hline
		Forms of $n$  & Length $n$ & Minimum Distance $d$ & References\\
		\hline
		$n\leq q+1$  &  $n\leq q+1$ & $2\leq d\leq \frac{n}{2}+1$ &  \cite{RefJ (2004) n<q+1(1),RefJ (2004) n<q+1(2)} \\
        \hline
         \multirow{4}{*}{$n\mid q^2+1$} & $n=q^2+1$ & $2\leq d\leq q+1$, $d\neq q$ & \cite{RefJ (2011) n=q^2+1(1),RefJ (2015) n=q^2+1(2),RefJ (2013) n=q^2+1(3),RefJ (2008) n=q^2+1(4),RefJ (2021) n=q^2+1(5)}\\
        \cline{2-4}
         \multirow{4}{*}{} & $n=\frac{q^2+1}{2}$, $q$ odd &  $2\leq d\leq q$, $d$ odd& \cite{RefJ (2013) n=q^2+1(3),RefJ (2014) L.jin} \\
        \cline{2-4}
        \multirow{4}{*}{} &$n=\frac{q^2+1}{5}$, $q\equiv \pm 3({\rm mod}\ 10)$ & $2\leq d\leq \frac{3q\pm 1}{5}$, $d$ even & \cite{RefJ (2014) kai,RefJ (2015) T.Zhang,RefJ (2016) L.Hu} \\
        \cline{2-4}
        \multirow{4}{*}{} &$n=\frac{q^2+1}{5}$,  $q\equiv \pm 2({\rm mod}\ 10)$ & $2\leq d\leq \frac{3q\mp 1}{5}$, $d$ odd & \cite{RefJ (2016) S.Li} \\
        \hline
        $(q+1)\mid (n-2)$ & $n=r(q+1)+2$, $1\leq r\leq q-1$ & \tabincell{c}{$2\leq d\leq r+2$,\\ $(p,r,d)\neq(2,q-1,q)$}& \cite{RefJ (2018)W. two} \\
        \hline
         $(q+1)\mid (n-1)$&  $n=r\frac{q^2-1}{h}+1$, $h|(q-1)$, $1\leq r\leq h$ & $2\leq d\leq r\frac{q-1}{h}+1$ & \cite{RefJ (2019)W. some} \\
        \hline
         \multirow{3}{*}{$(q+1)\mid n$}& $n=r\frac{q^2-1}{h}$, $h|(q-1)$, $1\leq r\leq h$ & $2\leq d\leq r\frac{q-1}{h}+1$ & \cite{RefJ (2017) X.Shi,RefJ (2018) X.Shi CC,RefJ (2019)W. some,RefJ (2023) R. Wan}  \\
        \cline{2-4}
         \multirow{3}{*}{} &$n=r\frac{q^2-1}{h}$, $2\mid h$, $h|(q-1)$, $1\leq r\leq \frac{h}{2}$ & $2\leq d\leq (\frac{h}{2}+r)\frac{q-1}{h}+1$ & \cite{RefJ (2014) kai,RefJ (2015) B.chen,RefJ (2016) X.He,RefJ (2017)Lem GRS,RefJ (2023) R. Wan} \\
        \cline{2-4}
         \multirow{3}{*}{}& $n=r\frac{q^2-1}{h}$, $2\mid h$, $h|(q-1)$, $\frac{h}{2}< r\leq h$ & $2\leq d\leq \lfloor \frac{h+r}{2}\rfloor\frac{q-1}{h}+1$ & \cite{RefJ (2017)Lem GRS,RefJ (2023) R. Wan} \\
        \hline
         \multirow{2}{*}{\tabincell{c}{$(q+1)\nmid n$,\\ $\frac{q+1}{2}\mid n$}}  & \tabincell{c}{$n=(\frac{h}{2}+1)\frac{q^2-1}{2h}$, $\frac{q-1}{h}=2\tau+1$,\\ $\tau\geq 1$, $\frac{h}{2}+1$ odd} & $2\leq d\leq (h+1)\frac{q-1}{2h}+\frac{1}{2}$ & \cite{RefJ (2019) F.Tian}\\
         \cline{2-4}
         \multirow{3}{*}{} & \tabincell{c}{$n=(\frac{h}{2}+2)\frac{q^2-1}{2h}$, $\frac{q-1}{h}=2\tau+1$,\\ $\tau\geq 1$, $\frac{h}{2}+2$ odd} & $2\leq d\leq (h+2)\frac{q-1}{2h}+1$ & \cite{RefJ (2019) F.Tian}\\
        \hline
         $p\mid n$ & \tabincell{c}{$n=tr^e$, $1\leq t\leq r$, \\ $q^2=r^m$, $1\leq e\leq m-1$}& $2\leq d\leq \lfloor  \frac{tr^e+q-1}{q+1} \rfloor+1$& \cite{RefJ (2008) n=q^2+1(4),RefJ (2018)W. two,RefJ (2020)W. hull}\\
        \hline
         $(q-1)\mid (n-1)$& \tabincell{c}{ $n=r\frac{q^2-1}{h}+1$, $h|(q+1)$,\\ $1\leq r< \min \{q,h\}$} & $2\leq d\leq \lfloor\frac{h+r}{2}\rfloor\frac{q+1}{h}$ & \cite{RefJ (2014) L.jin,RefJ (2016) X.He,RefJ (2017) L.Jin,RefJ (2019)W. some,RefJ (2023) R. Wan}  \\
        \hline
         $(q-1)\mid n$& $n=r\frac{q^2-1}{h}$, $h|(q+1)$, $1\leq r<h$  &  $2\leq d\leq \lceil\frac{h+r}{2}\rceil\frac{q+1}{h}-1$& \tabincell{c}{\cite{RefJ (2015) L.W,RefJ (2015) B.chen,RefJ (2016) X.He,RefJ (2017) X.Shi}\\ \cite{RefJ (2017) L.Jin,RefJ (2019)W. some,RefJ (2023) R. Wan}}\\
        \hline
          \multirow{3}{*}{Other lengths}  & \tabincell{c}{$n=r\frac{q^2-1}{s}+l\frac{q^2-1}{t}-rl\frac{q^2-1}{st}$,\\ odd $s\mid (q+1)$, even $t\mid (q-1)$, \\$r\leq s-1$, $l\leq t$, $rl\frac{q^2-1}{st}< q-1$ }
        & \tabincell{c}{$2\leq d\leq \min\{\lfloor \frac{s+r}{2}\rfloor\frac{q+1}{s}-1,$\\$ (\frac{t}{2}+1)\frac{q-1}{t}+1\}$} & \cite{RefJ (2020) X.Fang} \\
        \cline{2-4}
        \multirow{3}{*}{} & \tabincell{c}{$n=r\frac{q^2-1}{s}+l\frac{q^2-1}{t}-rl\frac{q^2-1}{st}+1$,\\ odd $s\mid (q+1)$, even $t\mid (q-1)$,\\ odd $r\leq s-1$, $l\leq t$, $rl\frac{q^2-1}{st}< q-1$ }
        &  \tabincell{c}{$2\leq d\leq \min\{ \frac{s+r}{2}\frac{q+1}{s},$\\$ (\frac{t}{2}+1)\frac{q-1}{t}+1\}$} & \cite{RefJ (2020) X.Fang} \\
        \cline{2-4}
          \multirow{3}{*}{} & \tabincell{c}{$n=r\frac{q^2-1}{s}+l\frac{q^2-1}{t}$, even $s\mid (q+1)$,\\ even $t\mid (q-1)$, $r\leq \frac{s}{2}$, $l\leq \frac{t}{2}$}
        & \tabincell{c}{ $2\leq d\leq \min\{ \lfloor \frac{s+r}{2}\rfloor\frac{q+1}{s}-1,$\\$ (\frac{t}{2}+1)\frac{q-1}{t}+1\}$} & \cite{RefJ (2020) X.Fang} \\
        \hline
	\end{tabular}}
 \begin{tablenotes}
     \footnotesize
    \item Note: $p$ is the characteristic of $q$.
    \item Note:  in \cite{RefJ (2016) X.He}, \cite{RefJ (2017)Lem GRS}, \cite{RefJ (2021) Ball determine} and \cite{RefJ (2024) Fang} many quantum MDS codes were also introduced.
    \end{tablenotes}
\end{center}
\end{table}

\begin{table}
\caption{Our new constructions of $[[n,n-2d+2,d]]_q$-quantum MDS codes}
\label{tab:3}
\begin{center}
\resizebox{\textwidth}{73mm}{
	\begin{tabular}{c|ccc}
		\hline
		Forms of $n$  & Length $n$ & Minimum Distance $d$ & References\\
        \hline
        $(q+1)\mid (n-1)$
      & \tabincell{c}{$n=r\frac{q^2-1}{h}+1$, $\frac{q-1}{h}=2\tau+1$,\\ $\tau\geq 1$, $\lceil \frac{h+2}{4}\rceil\leq r\leq \frac{h}{2}$} & $2\leq d\leq (\frac{h}{2}+\frac{1}{2})\frac{q-1}{h}+\frac{3}{2}$ & Theorem \ref{th 11} (1)\\
       \hline
        \multirow{2}{*}{\tabincell{c}{$(q+1)\nmid (n-1)$\\$\frac{q+1}{2}\mid (n-1)$}}
        & \tabincell{c}{$n=r\frac{q^2-1}{2h}+1$, $\frac{q-1}{h}=2\tau+1$,\\ $\tau\geq 1$, odd $\frac{h}{2}+1\leq r\leq h$} & $2\leq d\leq (h+1)\frac{q-1}{2h}+\frac{3}{2}$ & Theorem \ref{th 11} (1)\\
        \cline{2-4}
        \multirow{2}{*}{}
       & \tabincell{c}{$n=r\frac{q^2-1}{2h}+1$, $\frac{q-1}{h}=2\tau+1$,\\ $\tau\geq 1$, odd $h< r< 2h$} & $2\leq d\leq r\frac{q-1}{2h}+\frac{3}{2}$ & Theorem \ref{th 11} (2)\\
       \hline
       \multirow{2}{*}{\tabincell{c}{$(q+1)\nmid n$\\$\frac{q+1}{2}\mid n$}}
        & \tabincell{c}{$n=r\frac{q^2-1}{2h}$, $\frac{q-1}{h}=2\tau+1$,\\ $\tau\geq 1$, odd $\frac{h}{2}+1< r\leq h$} & $2\leq d\leq (h+2)\frac{q-1}{2h}+1$ &Theorem \ref{th 22}\\
        \cline{2-4}
        \multirow{2}{*}{}
       & \tabincell{c}{$n=r\frac{q^2-1}{2h}$, $\frac{q-1}{h}=2\tau+1$,\\ $\tau\geq 1$, odd $h< r< 2h$} & $2\leq d\leq r\frac{q-1}{2h}+\frac{1}{2}$ & Corollary \ref{cor 2} \\
       \hline
        \multirow{2}{*}{\tabincell{c}{$(q-1)\nmid (n-1)$\\$\frac{q-1}{2}\mid (n-1)$}}
        & \tabincell{c}{$n=r\frac{q^2-1}{2h}+1$, $\frac{q+1}{h}=2\tau+1$,\\ $\tau\geq 1$, odd $h< r\leq \frac{3h}{2}$} & $2\leq d\leq (r+1)\frac{q+1}{2h}$ & Theorem \ref{th 33} (1)\\
        \cline{2-4}
        \multirow{2}{*}{}
       & \tabincell{c}{$n=r\frac{q^2-1}{2h}+1$, $\frac{q+1}{h}=2\tau+1$,\\ $\tau\geq 1$, odd $\frac{3h}{2}< r< 2h$} & $2\leq d\leq r\frac{q+1}{2h}-\frac{1}{2}$ & Theorem  \ref{th 33} (2)\\
       \hline
         \multirow{3}{*}{\tabincell{c}{$(q-1)\nmid n$\\$\frac{q-1}{2}\mid n$}}
        & \tabincell{c}{$n=r\frac{q^2-1}{2h}$, $\frac{q+1}{h}=2\tau+1$,\\ $\tau\geq 1$, odd $\frac{h}{2}\leq r\leq h$} & $2\leq d\leq (h+1)\frac{q+1}{2h}-\frac{1}{2}$ & Theorem \ref{th 44} (1)\\
        \cline{2-4}
        \multirow{3}{*}{}
       & \tabincell{c}{$n=r\frac{q^2-1}{2h}$, $\frac{q+1}{h}=2\tau+1$,\\ $\tau\geq 1$, odd $h< r< \frac{3h}{2}$} & $2\leq d\leq (r+2)\frac{q+1}{2h}-\frac{1}{2}$ & Theorem \ref{th 44} (2) \\
       \cline{2-4}
       \multirow{3}{*}{}
       & \tabincell{c}{$n=r\frac{q^2-1}{2h}$, $\frac{q+1}{h}=2\tau+1$,\\ $\tau\geq 1$, odd $\frac{3h}{2}\leq r< 2h$} & $2\leq d\leq (r+1)\frac{q+1}{2h}-1$ & Theorem \ref{th 44} (3) \\
       \hline
	\end{tabular}}
\end{center}
\end{table}

\section{Conclusions}\label{sec6}

In this paper, by Hermitian self-orthogonal GRS codes
we construct some new classes of quantum MDS codes with good parameters (see Theorems \ref{th 11}, \ref{th 22}, \ref{th 33}, \ref{th 44}).
The length of these quantum MDS codes satisfies $n\equiv 0,1($mod$\,\frac{q\pm1}{2})$, which differs from previous conclusions.
All the quantum MDS codes we construct have large minimum distances which are greater than  $q/2+1$.

\section*{Acknowledgments}

This research was supported by the National Natural Science Foundation of China (No.U21A20428 and 12171134).

\section*{Data availability}

All data generated or analyzed during this study are included in this paper.

\section*{Conflict of Interest}

The authors declare that there is no possible conflict of interest.

\end{document}